\documentclass[11pt]{article}

\usepackage{amsthm}
\usepackage{amsmath,amssymb}

\usepackage{algorithmic}

\usepackage{graphicx,subfigure}

\usepackage{mathptmx}
\usepackage[margin=1in]{geometry}

\usepackage{algorithm,algorithmic}

\usepackage{float}
\newfloat{algorithm}{t}{lop}

\usepackage{placeins}

\usepackage[linkcolor=black,colorlinks=true,citecolor=black,filecolor=black]{hyperref}

\newtheorem{theorem}{Theorem}[section]

\newtheorem{lemma}[theorem]{Lemma}
\newtheorem{corollary}[theorem]{Corollary}
\newtheorem{definition}[theorem]{Definition}
\newtheorem{observation}[theorem]{Observation}

\newcommand{\R}{\mathbb{R}}
\newcommand{\E}{\ensuremath{\mathbf{E}}}

\newcommand{\sat}{\mathrm{sat}}
\newcommand{\vbl}{\mathrm{vbl}}
\newcommand{\ol}[1]{\bar{#1}}

\newcommand{\epsa}{{\epsilon_1}}
\newcommand{\epsb}{{\epsilon_2}}
\newcommand{\epsc}{{\epsilon_3}}
\newcommand{\dela}{{\Delta_1}}
\newcommand{\delb}{{\Delta_2}}
\newcommand{\pa}{{p^{*}}}

\newcommand{\succppsz}{\textrm{E}_\ppsz}
\newcommand{\succguess}{\textrm{E}_{\textrm{guess}}}

\newcommand{\poly}{\operatorname{poly}}
\newcommand{\ppsz}{\textsc{PPSZ}}
\newcommand{\gic}{\textsc{GetInd2Clauses}}
\newcommand{\onecc}{\textsc{OneCC}}
\newcommand{\sparsealg}{\textsc{Sparse}}
\newcommand{\densealg}{\textsc{Dense}}
\newcommand{\ws}{\textsc{Wahlstroem}}
\newcommand{\ppszimp}{\textsc{PPSZImproved}}
\newcommand{\ocnf}{{$(\leq 3)$-CNF}}
\newcommand{\cuocnf}{{1C-Unique $(\leq 3)$-CNF}}
\renewcommand{\H}{\mathrm{H}}
\newcommand{\getsuar}{\gets_{\text{u.a.r.}}}

\title{Breaking the PPSZ Barrier for Unique 3-SAT}

\author{
Timon Hertli\thanks{
Institute for Theoretical Computer Science,
Department of Computer Science,
ETH Z\"urich, 8092 Z\"urich, Switzerland (\texttt{timon.hertli@inf.ethz.ch})
}}
\begin{document}
\maketitle

\date



\begin{abstract}
The PPSZ algorithm by Paturi, Pudl\'ak, Saks, and Zane (FOCS 1998) is the fastest known algorithm for (Promise) Unique $k$-SAT. We give an improved algorithm with exponentially faster bounds for Unique $3$-SAT.

For uniquely satisfiable 3-CNF formulas, we do the following case distinction: We call a clause critical if exactly one literal is satisfied by the unique satisfying assignment. If a formula has many critical clauses, we observe that PPSZ by itself is already faster. If there are only few clauses allover, we use an algorithm by Wahlstr\"om (ESA 2005) that is faster than PPSZ in this case. Otherwise we have a formula with few critical and many non-critical clauses. Non-critical clauses have at least two literals satisfied; we show how to exploit this to improve PPSZ.

\end{abstract}



\section{Introduction}
The well-known problem $k$-SAT is NP-complete for $k\geq 3$. If P$\not=$NP, $k$-SAT does not have a polynomial time algorithm. For a CNF formula $F$ over $n$ variables, the naive approach of trying all satisfying assignments takes time $O(2^n\cdot \poly(|F|))$.
 Especially for $k=3$ much work has been put into finding so-called ``moderately exponential time'' algorithms running in time $O(2^{cn})$ for some $c<1$. In 1998, Paturi, Pudl\'ak, Saks, and Zane presented a randomized algorithm for 3-SAT that runs in time $O(1.364^n)$. Given the promise that the formula has at most one satisfying assignment (that problem is called Unique 3-SAT), a running time of $O(1.308^n)$ was shown. Both bounds were the best known when published. The running time of general $3$-SAT has been improved repeatedly (e.g.~\cite{schoning1999,it04}), until PPSZ was shown to run in time $O(1.308^n)$ for general 3-SAT~\cite{hertli11}.

Any further improvement of 3-SAT further also improves Unique 3-SAT; however that bound has not been improved upon since publication of the PPSZ algorithm. In this paper, we present a randomized algorithm for Unique 3-SAT with exponentially better bounds than what could be shown for PPSZ. Our algorithm builds on PPSZ and improves it by treating sparse and dense formulas differently.

A key concept of the PPSZ analysis is the so-called critical clause: We call a clause critical for a variable $x$ if exactly one literal is satisfied by this unique satisfying assignment, and that literal is over $x$. It is not hard to see that the uniqueness of the satisfying assignment implies that every variable has at least one critical clause. If some variables have strictly \emph{more} than one critical clause, then we will give a straightforward proof that PPSZ by itself is faster already. Hence the bottleneck of PPSZ is when every variable has \emph{exactly} one critical clause, and in total there are exactly $n$ critical clauses.

Given a formula with exactly $n$ critical clauses, consider how many other (non-crtical) clauses there are. If there are few, we use an algorithm by Wahlstr\"om~\cite{wahlstroem05} that is faster than PPSZ for formulas with few clauses allover. If there are many non-critical clauses we use the following fact: A non-critical clause has two or more satisfied literals (w.r.t.\ unique satisfying assignment); so after removing a literal, the remaining 2-clause is still satisfied. We will exploit this to improve PPSZ.

An remaining problem is the case if only very few (i.e.\ sublinearly many) variables have more than one critical clause or appear in many (non-critical) clauses. In this case, we would get only a subexponential improvement. A significant part of our algorithm deals with this problem.

\subsection{Notation}
We use the notational framework introduced in \cite{welzl05}. 
Let $V$ be a finite set of propositional \emph{variables}. A \emph{literal} $u$ over $x\in V$ is a variable $x$ or a negated variable $\ol{x}$. If $u=\ol{x}$, then $\ol{u}$, the negation of $u$, is defined as $x$. We mostly use $x,y,z$ for variables and $u,v,w$ for literals.
We assume that all lite\-rals are distinct. A \emph{clause} over $V$ is a finite set of lite\-rals over pairwise distinct variables from $V$. By $\vbl(C)$ we denote the set of variables that occur in $C$, i.e.\ $\{x\in V\mid x\in C \vee \ol{x}\in C\}$. $C$ is a $k$-clause if $|C|=k$ and it is a $(\leq k)$-clause if $|C|\leq k$.
A formula in \emph{CNF} (Conjunctive Normal Form) $F$ over $V$ is a finite set of clauses over $V$. We define $\vbl(F):=\bigcup_{C\in F}\vbl(C)$. 
$F$ is a k-CNF formula (a $(\leq k)$-CNF formula) if all clauses of $F$ are $k$-clauses ($(\leq k)$-clauses).
A (truth) \emph{assignment} on $V$ is a function $\alpha : V \rightarrow
\{0,1\}$ which assigns a Boolean value to each variable. 
$\alpha$ extends to negated variables by letting $\alpha(\ol{x}):=1-\alpha(x)$.
A literal $u$ is \emph{satisfied by} $\alpha$ if $\alpha(u)=1$. A clause is \emph{satisfied by} $\alpha$ if it
contains a satisfied literal and a formula is \emph{satisfied by}
$\alpha$ if all of its clauses are. A formula is \emph{satisfiable} if
there exists a satisfying truth assignment to its variables. A formula that is not satisfiable is called \emph{unsatisfiable}.
Given a CNF formula $F$, we denote by $\sat(F)$ the set of assignments on $\vbl(F)$ that
satisfy $F$. $k$\emph{-SAT} is the decision problem of deciding if a $(\leq k)$-CNF formula has a satisfying assignment.

If $F$ is a CNF formula and $x \in \vbl(F)$, we write $F^{[x\mapsto 1]}$ (analogously $F^{[x\mapsto 0]}$) for the formula arising from removing all clauses containing $x$ and truncating all clauses containing $\ol{x}$ to their remaining literals. This corresponds to assigning $x$ to $1$ (or $0$) in $F$ and removing trivially satifsied clauses. We call assignments $\alpha$ on $V$ and $\beta$ and $W$ \emph{consistent} if $\alpha(x)=\beta(x)$ for all $x\in V\cap W$. If $\alpha$ is an assignment on $V$ and $W\subseteq V$, we denote by $\alpha|_W$ the assignment on $W$ with $\alpha|_W(x)=\alpha(x)$ for $x\in W$.
If $\gamma=\{x\mapsto 0,y\mapsto 1,\dots\}$, we write $F^{[\gamma]}$ as a shorthand for $F^{[x\mapsto 0][y\mapsto 1]\dots}$, the \emph{restriction} of F to $\gamma$.

%
%
%
For a set $W$, we denote by $x\getsuar W$ choosing an element $x$ u.a.r. (uniformly at random).
Unless otherwise stated, all random
choices are mutually independent.
We denote by $\log$ the logarithm to the base 2. For the logarithm to
the base $e$, we write $\ln$. 
By $\poly(n)$ we denote a polynomial factor depending on $n$.
We use the following convention if no confusion arises: When $F$ is a CNF formula, we denote by $V$ its variables and by $n$ the number of variables of $F$, i.e.\ $V:=\vbl(F)$ and $n:=|\vbl(F)|$. By $o(1)$ we denote a quantity dependent on $n$ going to $0$ with $n\to\infty$.
\subsection{Previous Work}
\begin{definition}
(Promise) Unique 3-SAT is the following promise problem: Given a {\ocnf} with at most one satisfying assignment, decide if it is satisfiable or unsatisfiable.

A randomized algorithm for Unique 3-SAT is an algorithm that, for a uniquely satisfiable {\ocnf} formula returns the satisfying assignment with probability $\frac{1}{2}$.
\end{definition}
Note that if the formula is not satisfiable, there is no satisfying assignment, and the algorithm cannot erroneously find one. Hence the error is one-sided and we don't have to care about unsatisfiable formulas. 

The PPSZ algorithm~\cite{ppsz} is a randomized algorithm for Unique 3-SAT running in time $O(1.308^n)$. The precise bound is as follows:
\begin{definition}
Let $S:=\int_{0}^{1}\left(1-\min\{1,\frac{r^2}{(1-r)^2}\}\right)dr=2\ln 2 - 1$.
\end{definition}
\begin{theorem}[\cite{ppsz}]
\label{thm.ppsz}
There exists a randomized algorithm (called $\ppsz$) for Unique 3-SAT running in time $2^{(S+o(1))n}$.
\end{theorem}
Note that $0.3862<S<0.3863$ and $2^{S}<1.308$.
\subsection{Our Contribution}
For Unqiue 3-SAT, we get time bounds exponentially better than PPSZ:
\begin{theorem}
\label{thm.main}
There exists a randomized algorithm for Unique 3-SAT running in time $2^{(S-\epsb+o(1))n}$ where $\epsb=10^{-24}$.
\end{theorem}

In Section~\ref{sec.ppsz}, we review the PPSZ algorithm. In Section~\ref{sec.1cc}, we show that the worst case for PPSZ occurs when every variable has exactly one critical 3-clause; this case we improve in Section~\ref{sec.i2c}. In Section~\ref{sec.con}, we pose open problems that arise.
\section{The PPSZ Algorithm}
\label{sec.ppsz}
In this section we review the PPSZ algorithm~\cite{ppsz}, summarized in Algorithm~\ref{alg.ppsz}. We need to adapt some statements slightly. For the straightforward but technical proofs we refer the reader to the appendix. The following two definitions are used to state the PPSZ algorithm.

\begin{definition}
A CNF formula $F$ \emph{$D$-implies} a literal $u$ if there exists a subformula $G\subseteq F$ with $|G|\leq D$ and all satisfying assignments of $G$ set $u$ to $1$.
\end{definition}

In a random permutation, the positions of two elements are not independent. To overcome this, placements were defined. They can be seen as continuous permutations with the nice property that the places of different elements are independent.
\begin{definition}[\cite{ppsz}]
A placement on $V$ is a mapping $V\to[0,1]$. A random placement is obtained by choosing for every $x\in V$ $\pi(x)$ uniformly at random from $[0,1]$, independently.
\end{definition}
\begin{observation}
By symmetry and as ties happen with probability $0$, ordering $V$ according to a random placement gives a permutation distributed the same as a permutation drawn uniformly at random from the set of all permutations on $V$.
\end{observation}

\begin{algorithm}
\caption{$\ppsz($CNF formula $F)$}
\label{alg.ppsz}
\begin{algorithmic}
  \STATE $V\gets \vbl(F)$; $n\gets |V|$
  \STATE Choose $\beta$ u.a.r.\ from all assignments on $V$ 
  \STATE Choose $\pi: V\to [0,1]$ as a random placement of $V$ 
  \STATE Let $\alpha$ be a partial assignment on $V$, initially empty
  \FOR {$x\in V$, in ascending order of $\pi(x)$}
  \STATE \textbf{if} $F$ $(\log n)$-implies $x$ or $\bar{x}$, set $\alpha(x)$ to satisfy this literal
  \STATE \textbf{otherwise} $\alpha(x)\gets\beta(x)$
  \COMMENT {guess $\alpha(x)$ u.a.r.}
   \STATE $F\gets F^{[x\mapsto \alpha(x)]}$
  \ENDFOR
  \RETURN $\alpha$
\end{algorithmic}
\end{algorithm}

The analysis of PPSZ builds on the concept of forced and guessed variables:

\begin{definition}
If in $\ppsz$, $\alpha(x)$ is assigned $0$ or $1$ because of $D$-implication, we call $x$ \emph{forced}. Otherwise (if $\alpha(x)$ is set to $\beta(x)$), we call $x$ \emph{guessed}.
\end{definition}


The following lemma from~\cite{ppsz} relates the expected number of guessed variables to the success probability (the proof is by an induction argument and Jensen's inequality).

\begin{lemma}[\cite{ppsz}]
\label{lem.ppsz}
Let $F$ be a satisfiable {\ocnf}, let $\alpha^*$ be a satisfying assignment. Let $G(\pi)$ be the expected number of guessed variables conditioned on $\beta=\alpha^*$ depending on $\pi$. Then $\ppsz(F)$ returns $\alpha^*$ with probability at least $\E_{\pi}[2^{-G(\pi)}]\geq 2^{\E_{\pi}[-G(\pi)]}$.
\end{lemma}
Remember that $S:=\int_{0}^{1}\left(1-\min\{1,\frac{r^2}{(1-r)^2}\}\right)dr=2\ln 2 - 1$, which corresponds to the probability that a variable is guessed. We define $S_p$ where the integral starts from $p$ instead of $0$; this corresponds to the probability that a variable has place at least $p$ and is guessed.
\begin{definition}
Let $S_p:=\int_{p}^{1}\left(1-\min\{1,\frac{r^2}{(1-r)^2}\}\right)dr$.
\end{definition}
\begin{observation}
\label{obs.sp}
For $p\leq \frac{1}{2}$, $S_p=S-p+\int_{0}^{p}\frac{r^2}{(1-r)^2}dr$.
\end{observation}
In the appendix, we derive from~\cite{ppsz} the following:
\begingroup
\def\thetheorem{\ref{cor.expguessed}}
\begin{corollary}
\sloppypar{
Let $F$ a {\ocnf} with unique satisfying assignment $\alpha$. Then in PPSZ($F$) conditioned on $\beta=\alpha$, the expected number of guessed variables is at most $(S+o(1))n$.

Furthermore, suppose we pick every variable of $F$ with probability $p$, independently, and let $V_p$ be the resulting set. Then in PPSZ($F$) conditioned on $\beta=\alpha$, the expected number of guessed variables is at most $(S_p+o(1))n$.
}
\end{corollary}
\addtocounter{theorem}{-1}
\endgroup
By Lemma~\ref{lem.ppsz}, we have the following corollary:
\begin{corollary}
\label{cor.sp}
Let $F$ a {\ocnf} with unique satisfying assignment $\alpha$. Then the probability that $\ppsz(F)$ returns $\alpha$ is at least $2^{(-S-o(1))n}$.

Furthermore, suppose we pick every variable of $F$ with probability $p$, independently, and let $V_p$ be the resulting set. Then the expected $\log$ of the probability (over the choice of $V_p$) that $\ppsz(F^{[\alpha|_{V_p}]})$ returns $\alpha|_{V\setminus V_p}$ is at least
$(-S_p-o(1))n$.
\end{corollary}
The first statement is actually what is shown in~\cite{ppsz}, and the second statement is a direct consequence. We need this later when we replace PPSZ by a different algorithm on variables with place at most $p$. It is easily seen that for a $(\leq 3)$-CNF $F$, $\ppsz(F)$ runs in time $2^{o(n)}$. Hence by a standard repetition argument, $\ppsz$ gives us an algorithm finding an assignment in time $2^{(S+o(1))n}$ and we (re-)proved Theorem~\ref{thm.ppsz}.

\section{Reducing to One Critical Clause per Variable}
\label{sec.1cc}
In this section we show that an exponential improvement for the case where every variable has exactly one critical clause gives an exponential improvement for Unique 3-SAT.
%
%
\begin{definition}[\cite{ppsz}]
Let $F$ be a CNF formula satisfied by $\alpha$.
We call a clause $C$ \emph{critical} for $x$ (w.r.t.\ $\alpha$) if $\alpha$ satisfies exactly one literal of $C$, and this literal is over $x$.
\end{definition}
\begin{definition}
A {\cuocnf} is a uniquely satisfiable {\ocnf} where every variable has at most one critical clause. Call the corresponding promise problem 1C-Unique 3-SAT.
\end{definition}

All formulas we consider have a unique satisfying assignment; critical clauses will be always w.r.t.\ that. First we show that a variables with more than one critical clause are guessed less often; giving an exponential improvement for formulas with a linear number of such variables. A similar statement for shorter critical clauses is required in the next section.

\begin{lemma}
\label{lem.morethanone}
Let $F$ be a {\ocnf} uniquely satisfied by $\alpha$. A variable $x$ with at least two critical clauses (w.r.t.\ $\alpha$) is guessed given $\beta=\alpha$ with probability at most $S - 0.0014 + o(1)$. Furthermore, a variable $x$ with a critical $(\leq 2)$-clause is guessed with probability at most $S-0.035 + o(1)$
\end{lemma}

\begin{proof}
Suppose $\pi(x)=r$.
Let $C_1$ and $C_2$ be two critical clauses of $x$. If $C_1$ and $C_2$ share no variable besides $x$, then the probability that $x$ is forced is at least $2r^2-r^4$ by the inclusion-exclusion principle. If $C_1$ and $C_2$ share one variable besides $x$, then the probability that $x$ is forced is at least $2r^2-r^3$ (which is smaller than $2r^2-r^4$. $C_1$ and $C_2$ cannot share two varibles besides $x$: in that case $C_1=C_2$, as being a critical clause for $x$ w.r.t.\ $\alpha$ predetermines the polarity of the literals. Intutiviely, if $r$ is small, then $2r^2-r^3$ is almost twice as large as $\frac{r^2}{(1-r)^2}$; therefore in this area the additional clause helps us and the overall forcing probability increases. For a critical $(\leq 2)$-clause the argument is analogous. Here, the probability that $x$ is forced given place $r$ is at least $r$. The statement follows now by integration using the dominated convergence theorem, see appendix~\ref{subs.integrals}.
\end{proof}

\begin{corollary}
\label{cor.morecc}
\label{cor.2cc}
Let $F$ be a {\ocnf} formula uniquely satisfied by $\alpha$. 
If $\Delta n$ variables of $F$ have two critical clause, PPSZ finds $\alpha$ with probability at least
$2^{-(S-0.0014\Delta + o(1))n}.$

\sloppy{
 If  $\Delta n$ variables of $F$ have a critical $(\leq 2)$-clause clause, PPSZ finds $\alpha$ with probability at least
 $2^{-(S-0.035\Delta + o(1))n}.$
 }
\end{corollary}

If there are only few variables (less than $\dela n$) with one critical clause, we can find and guess them by brute-force. If we choose $\dela$ small enough, any exponential improvement for 1C-Unique 3-SAT gives a (diminished) exponential improvement to Unique 3-SAT.
To bound the number of subsets of size $\dela n$, we define the binary entropy and use a well-known upper bound to the binomial coefficient.
\begin{definition}
For $p\in [0,1]$,
$H(p):=-p\log p -(1-p)\log (1-p)$ ($0\log 0:=0$).
\end{definition}
\begin{lemma}[Chapter 10, Corollary 9 of~\cite{ms77}]
If $pn$ is an integer, then
\label{lem.entropy}
$\binom{n}{pn}\leq 2^{H(p)n}.$
\end{lemma}

We will manily prove that we have \emph{some} exponential improvement. The claimed numbers are straightforward to check by inserting the values from the following table.
%
%
\begin{center}
\begin{tabular}{|c|c|p{0.72\textwidth}|}
\hline
name&value&description\\
\hline
$\epsa$&$10^{-19}$&improvement in 1C-Unique 3-SAT\\
\hline
$\epsb$&$10^{-24}$&improvement in Unique 3-SAT\\
\hline
$\dela$&$10^{-21}$&threshold fraction of vars. with more than 1 crit. clause\\
\hline
$\delb$&$6\cdot 10^{-5}$&$\delb n$ is the amount of variables for $\delb$-sparse and $\delb$-dense\\
\hline
$\epsc$&$10^{-3}$&exponential savings on repetitions if $F$ is $\delb$-sparse\\
\hline
$\pa$&$8\cdot 10^{-7}$&prob. that a var. is assigned using indep. $2$-clauses instead of $\ppsz$\\
\hline
\end{tabular}
\end{center}

\begin{lemma}
\label{lem.red}
If there is a randomized algorithm $\onecc(F)$ solving 1C-Unique 3-SAT in time $2^{(S-\epsa+o(1))n}$ for $\epsa>0$, then there is a randomized algorithm (Algorithm~\ref{alg.ppszimp}) solving Unique 3-SAT in time $2^{(S-\epsb+o(1))n}$ for some $\epsb>0$.
\end{lemma}
\begin{proof}
Let $F$ be a {\ocnf} uniquely satisfied by $\alpha$. Let $c(F)$ be the number of variables of $F$ with more than one critical clause.
If $c(F)\geq \dela n$, PPSZ is faster by Corollary~\ref{cor.morecc}. If $c(F)=0$, we can use $\onecc(F)$. 

However, what if $0<c(F)<\dela n$? In that case, we get rid of these variables by brute-force: For all $\lfloor\dela n\rfloor$-subsets $W$ of variables and for all $2^{\lfloor\dela n\rfloor}$ possible assignments $\alpha'$ on $W$, we try $\onecc(F^{[\alpha']})$.
For one such $\alpha'$, we have $F^{[\alpha']}$ satisfiable and $c(F)=0$; namely if $W$ includes all variables with multiple critical clauses and $\alpha'$ is compatible with $\alpha$. This is because fixing variables according to $\alpha$ does not produce new critical clauses w.r.t.\ $\alpha$.

There are $\binom{n}{\lfloor\dela n\rfloor}$ subsets of size $\lfloor\dela n\rfloor$ of the variables of $F$, each with $2^{\lfloor\dela n\rfloor}$ possible assignments. As $\binom{n}{\lfloor\dela n\rfloor}\leq 2^{\H(\dela) n}$ (Lemma~\ref{lem.entropy}), we invoke $\onecc(F^{[\alpha']})$ at most $2^{(\dela + \H(\dela))n}$ times. Setting $\dela$ small enough such that $\dela+\H(\dela)<\epsa$ retains an improvement for Unique 3-SAT.


\end{proof}
\begin{algorithm}
\caption{$\ppszimp($CNF formula $F)$}
\label{alg.ppszimp}
\begin{algorithmic}
\STATE repeat $\ppsz(F)$ $2^{(S-\epsb)n}$ times, return if a satisfying assignment has been found
\STATE for all subsets $W$ of size $\lfloor\dela n\rfloor$ and all assignments $\alpha'$ on $W$, try $\onecc(F^{[\alpha']})$
\end{algorithmic}
\end{algorithm}
\section{Using One Critical Clause per Variable}
\label{sec.i2c}
In this section we give an exponential improvement for 1C-Unique 3-SAT.
\begin{theorem}
\label{thm.onecc}
Given a {\cuocnf} on $n$ variables, \onecc($F$) runs in expected time $2^{(S-\epsa+o(1))n}$ and finds the satisfying assignment with probability $2^{-o(n)}$.
\end{theorem}
Obtaining a randomized algorithm using $2^{o(n)}$ independent repetitions and Markov's inequality is straightforward.
\begin{corollary}
\sloppypar{
There exists a randomized algorithm for 1C-Unique 3-SAT running in time $2^{(S-\epsa+o(1))n}$.}
\end{corollary}
Together with Lemma~\ref{lem.red} this immediately implies Theorem~\ref{thm.main}.
%
%
We obtain the improvement by doing a case distinction into sparse and dense formulas, as defined now:
\begin{definition}
For a CNF formula $F$ and a variable $x$, the \emph{degree} of $x$ in $F$, $\deg(F,x)$ is defined to be the number of clauses in $F$ that contain the variable $x$. The \emph{3-clause degree} of $x$ in $F$, $\deg_3(F,x)$ is defined to be the number of 3-clauses in $F$ that contain the variable $x$.

For a set of variables $W$, denote by $F\setminus W$ the part of $F$ \emph{independent} of $W$ that consists of the clauses of $F$ that do not contain variables of $W$.

We say that $F$ is \emph{$\Delta$-sparse} if there exists a set $W$ of at most $\Delta n$ variables such that $F\setminus W$ has maximum 3-clause degree $4$. We say that $F$ is \emph{$\Delta$-dense} otherwise. 
\end{definition}
We will show that for $\delb$ small enough, we get an improvement for $\delb$-sparse {\cuocnf} formulas. On the other hand, for any $\delb$ we will get an improvement for  $\delb$-dense {\cuocnf} formulas. In the sparse case we can fix by brute force a small set of variables to obtain a formula with few 3-clauses. We need to deal with the $(\leq 2)$-clauses and then use an algorithm from Wahlstr\"om for CNF formulas with few clauses.

\begin{algorithm}
\caption{$\onecc(${\ocnf} $F)$}
\label{alg.onecc}
\begin{algorithmic}
\STATE try $\densealg(F)$
\STATE try $\sparsealg(F)$
\end{algorithmic}
\end{algorithm}
\begin{algorithm}
\caption{\gic(\ocnf $F$)}
\begin{algorithmic}
\STATE {}
\COMMENT {for the analysis, $F$ is considered to be $\delb$-dense; the procedure might fail otherwise}
\STATE $F_3\gets \{C\in F\mid |C|=3\}$, $F_2\gets \{\}$
\FOR {$\lceil \delb n \rceil$ times}
\STATE let $x$ be a variable with $\deg_3(F,x)\geq 5$ (return failure if no such variable exists)
\STATE Choose $C$ u.a.r.\ from all of $F$ with $x\in\vbl(C)$.
\STATE $l\gets$ literal of $C$ over $x$; $C_2\gets C\setminus l$
\STATE $F_2\gets F_2\cup C_2$
\STATE {}
\COMMENT {remove all clauses of $F_3$ sharing variables with $C_2$}
\STATE $F_3\gets \{C_3\in F_3\mid \vbl(C')\cap \vbl(C_2)=\emptyset\}$
\ENDFOR
\RETURN $F_2$
\end{algorithmic}
\end{algorithm}
\begin{algorithm}
\caption{$\densealg(${\ocnf} $F)$}
\label{alg.densealg}
\begin{algorithmic}
\STATE $F_2\gets $\gic($F$)
\FOR {$2^{(S-\epsa)n}$ times}
\STATE $V_{\pa}\gets $ pick each $x\in\vbl(F)$ with probability $\pa$
\STATE $\alpha'\gets\{\}$
\FOR {$C_2\in F_2$}
\IF {$\vbl(C_2)\subseteq V_p$}
\STATE Let $\{u,v\}=C_2$
\STATE $(\alpha'(u),\alpha'(v))\gets\begin{cases}
(0,0),&\text{ with probability }\frac{3}{15}\\
(0,1),(1,0),(1,1),&\text{ with probability }\frac{4}{15}\text{ each}\\
\end{cases}$
\ENDIF
\ENDFOR
\STATE \textbf{for all} $x\in V_p$, if $\alpha'(x)$ is not defined yet let $\alpha'(x)\getsuar \{0,1\}$
\STATE $\ppsz(F^{[\alpha']})$; if a satisfying assignment $\alpha$ has been found, return $\alpha\cup\alpha'$
\ENDFOR
\end{algorithmic}
\end{algorithm}
%
%
%
%

\subsection{Dense Case}
First we show the improvement for any $\delb$-dense {\cuocnf}. $\delb$-density means that even after ignoring all clauses over any $ \delb n$ variables, a variable with 3-clause degree of at least 5 remains. The crucial idea is that for a variable $x$ with 3-clause degree of at least 5, picking one occurence of $x$ u.a.r.\ and removing it gives a 2-clause satisfied (by the unique satisfying assignment) with probability at least $\frac{4}{5}$. The only way a non-satisfied $2$-clause can arise is if the 3-clause $x$ was deleted from was critical for $x$. However we assumed that there is at most one critical clause for $x$. 

Repeating such deletions and ignoring all 3-clauses sharing variables with the produced 2-clauses, as in listed in \gic($F$), gives us the following:
\begin{observation}
\label{obs.ind2}
For a $\delb$-dense {\cuocnf} $F$, \gic($F$) returns a set of $\lceil \frac{1}{2}\delb n \rceil$ independent $2$-clauses, each satisfied (by the unique satisfying  assignment of $F$) independently with probability $\frac{4}{5}$.
\end{observation}

As a random 2-clause is satisfied with probability $\frac{3}{4}$ by a specific assignment, this set of 2-clauses gives us nontrivial information about the unique satisfying asignment. Now we show how to use these 2-clauses to improve PPSZ:
\begin{lemma}
\label{lem.manyok}
\sloppy{
Let $F$ be a $\delb$-dense {\cuocnf} for some $\delb>0$. Then there exists an algorithm ($\densealg(F)$) runing in time $2^{(S-\epsa+o(1))n}$ for $\epsa>0$ and returning the satisfying assignment $\alpha$ of $F$ with probability $2^{-o(n)}.$
}
\end{lemma}
\begin{proof}
First we give some intuition. For variables that occur late in PPSZ, the probability of being forced is large (being almost $1$ in the second half). However for variables that come at the beginning, the probability is very small; a variable $x$ at place $p$ is forced (in the worst case) with probability $\Theta(p^2)$ for $p\to 0$, hence we expect $\Theta(p^3 n)$ forced variables among the first $pn$ variables in total. 

However, a $2$-clause that is satisfied by $\alpha$ with probability $\frac{4}{5}$ can be used to guess both variables in a better way than uniform, giving constant savings in random bits required. For $\Theta(n)$ such $2$-clauses, we expect $\Theta(p^2 n)$ of them to have both variables among the first $pn$ variables. For each 2-clause we have some nontrivial information; intuitively we save around $0.01$ bits. In total we save $\Theta(p^2 n)$ bits among the first $pn$ variables, which is better than PPSZ for small enough $p$.

Formally, let $V_\pa$ be a random set of variables, where each variable of $V$ is added to $V_\pa$ with probability $\pa$. On $V_\pa$, we replace PPSZ by our improved guessing; on the remaining variables $V\setminus V_\pa$ we run PPSZ as usual. Let $\succguess$ be the event that the guessing on $V_\pa$ (to be defined later) finds $\alpha|_{V_\pa}$. Let $\succppsz$ be the event that PPSZ($F^{[\alpha|_{V_\pa}]}$) finds $\alpha|_{V\setminus V_\pa}$. Observe that for a fixed $V_\pa$, $\succguess$ and $\succppsz$ are independent. Hence we can write the overall probability to find $\alpha$ (call it $p_s$) as an expectation over $V_\pa$:
\begin{align*}
p_s&=E_{V_\pa}[\Pr(\succguess\cap\succppsz\vert V_\pa)]\\
&=E_{V_\pa}[\Pr(\succguess\vert V_\pa)\Pr(\succppsz\vert V_\pa)]\\
&=E_{V_\pa}[2^{\log\Pr(\succguess\vert V_\pa)+\log\Pr(\succppsz\vert V_\pa)}]\\
&\geq 2^{E_{V_\pa}[\log\Pr(\succguess\vert V_\pa)+\log\Pr(\succppsz\vert V_\pa)]}\\
&=2^{E_{V_\pa}[\log\Pr(\succguess\vert V_\pa)]+E_{V_\pa}[\log\Pr(\succppsz\vert V_\pa)]},
\end{align*}
where in the last two steps we used Jensen's inequality and linearity of expectation.

\sloppy{
By Corollary~\ref{cor.sp}, $E_{V_\pa}[\log\Pr(\succppsz)]=(-S_p+o(1))n$. We now define the guessing and analyze $E_{V_\pa}[\log\Pr(\succguess)]$ (see Algorithm~\ref{alg.densealg} as a reference):
}
By Observation~\ref{obs.ind2} we obtain a set of $\lceil \frac{1}{2}\delb n \rceil$ independent $2$-clauses $F_2$, each satisfied (by $\alpha$) independently with probability $\frac{4}{5}$. In the following we assume that $F_2$ has at least a $\frac{4}{5}$-fraction of satisfied 2-clauses as this happens with constant probability (for a proof, see e.g.~\cite{hamza99}) and we only need to show subexponential success probability.

Using the set of 2-clauses $F_2$, we choose an assignment $\alpha'$ on $V_\pa$ as follows: For every clause $C_2$ in $F_2$ completely over $V_\pa$ choose an assignment on both of its variables: with probability $\frac{1}{5}$ such that $C_2$ is violated, and with probability $\frac{4}{15}$ each on of the three assignments that satisfy $C_2$. Afterwards, guess any remaining variable of $V_\pa$ u.a.r.\ from $\{0,1\}$.
%


Given $V_\pa$, let $m_0$ be the number of clauses of $F_2$ completely over $V_\pa$ \emph{not satisfied} by $\alpha$. Let $m_1$ be the number of clauses of $F_2$ completely over $V_\pa$ \emph{satisfied} by $\alpha$. Then
\[\Pr(\succguess\vert V_\pa)=\left(\frac{1}{2}\right)^{|V_\pa|-2m_0-2m_1}\left(\frac{1}{5}\right)^{m_0}\left(\frac{4}{15}\right)^{m_1}.\]
This is seen as follows: For any variable for which no clause in $C_2$ is completely over $V_\pa$, we guess uniformly at random and so correctly with probability $\frac{1}{2}$. For any clause $C_2$ which is completely over $V_\pa$, we violate the clause with probability $\frac{1}{5}$, and choose a non-violating assignment with probability $\frac{4}{5}$. For any clause not satisfied by $\alpha$, we hence set both variables according to $\alpha$ with probability $\frac{1}{5}$. For any clause satisfied by $\alpha$, we set both variables according to $\alpha$ with probability $\frac{4}{15}$, as we have to pick the right one of the three assignments that satisfy $C_2$. As $E[V_\pa]=\pa n$, $E[m_0]\leq\frac{1}{5}\pa^2\lceil\delb n\rceil$, $E[m_1]\geq\frac{4}{5}\pa^2\lceil\delb n\rceil$, $E[m_0+m_1]=\pa^2\lceil\delb n\rceil$, we have
\begin{align*}E[\log \Pr(\succguess\vert V_\pa)]&=-E[V_\pa-2m_0-2m_1]+\log\left(\frac{1}{5}\right)E[m_0]+\log\left(\frac{4}{15}\right)E[m_1]\\
&\geq -\pa n+\pa^2\lceil\delb n\rceil\left(2+\log\left(\frac{1}{5}\right)\frac{1}{5}+ \log\left(\frac{4}{15}\right)\frac{4}{5}\right).
\end{align*}
The inequality follows from the observations and $\log\left(\frac{4}{15}\right)\geq\log\left(\frac{1}{5}\right)$.
One can calculate $2+\log\left(\frac{1}{5}\right)\frac{1}{5}+ \log\left(\frac{4}{15}\right)\frac{4}{5}\geq 0.01$, corresponding to the fact that a four-valued random variable where one value occurs with probability at most $\frac{1}{5}$ has entropy at most $1.99$.

Hence by our calculations and Observation~\ref{obs.sp} (to evaluate $S_p$), we have
\begin{align*}\frac{1}{n}\log p_s
&\geq -S+\pa-\int_{0}^\pa \frac{r^2}{(1-r)^2}dr + o(1) - \pa + \delb \pa^2 \cdot 0.01\\
&=-S-\int_{0}^\pa \frac{r^2}{(1-r)^2}dr+\delb \pa^2 \cdot 0.01+o(1).
\end{align*}
This gives an improvement over $\ppsz$ of $-\int_{0}^\pa \frac{r^2}{(1-r)^2}dr+\delb \pa^2\cdot 0.01+o(1)$. The first term corresponds to the savings $\ppsz$ would have, the second term corresponds to the savings we have in our modified guessing.
Observe that for small $\pa$, the integral evaluates to $\Theta(\pa^3)$, but the second term is $\Theta(\pa^2)$. Hence choosing $\pa$ small enough gives an improvement.
\end{proof}

\begin{algorithm}
\caption{$\sparsealg(${\ocnf} $F)$}
\begin{algorithmic}
\STATE repeat the following $2^{(S-\epsc)n}$ times:
\FOR {all subsets $W$ of size $\lfloor\delb n\rfloor$ and all assignments $\alpha'$ on $W$}
\STATE $F'\gets F^{[\alpha']}$
\WHILE {no satisfying assignment found}
\STATE try $\ppsz(F^{[\alpha']})$ 
\STATE $F'_2\gets \{C\in F'\mid |C|\leq 2\}$
\IF {$|F'_2|\leq \frac{1}{10} |\vbl(F')|$}
\STATE with probability $2^{-(S-0.015)|\vbl(F')|}$, run \ws($F'$)
\ENDIF
\STATE {}
\COMMENT {set all literals in a uniform $(\leq 2)$-clause to $1$}
\STATE $C'\getsuar F'_2$, if $F'_2=\{\}$ return failure
\FOR {$l\in C'$}
\STATE $F'\gets F'^{[l\mapsto 1]}$
\ENDFOR
\ENDWHILE
\ENDFOR
\end{algorithmic}
\end{algorithm}

\subsection{Sparse Case}
Now we show that if $\delb>0$ is small enough we get an improvement for a $\delb$-sparse {\cuocnf}. For this, we need the following theorem by Wahlstr\"om:
\begin{theorem}[\cite{wahlstroem05}]
\label{thm.ws}
Let $F$ be a CNF formula with average degree at most $4.2$ where we count degree $1$ as $2$ instead. Then satisfiability of $F$ can be decided in time $O(2^{0.371n})\leq 2^{(S-0.015+o(1))n}$. Denote this algorithm by $\ws(F)$.
\end{theorem}
\begin{lemma}
\label{lem.fewok}
\sloppy{
Let $F$ be a $\delb$-sparse {\cuocnf}. For $\delb$ small enough, there exists an algorithm ($\sparsealg(F)$) running in expected time $2^{(S-\epsc+o(1))n}$ for $\epsc>0$ and finding the satisfying assignment $\alpha$ of $F$ with probability $2^{-o(n)}$.
}
%
\end{lemma}
\begin{proof}
Similar to Section~\ref{sec.1cc}, we first check by brute-force all subsets $W$ of $\lfloor\delb n\rfloor$ variables and all possible assignments $\alpha'$ of $W$; by definition of $\delb$-sparse for some $W$, the part of $F$ independent of $W$ (i.e.\ $F\setminus W$) has maximum 3-clause degree $4$. If furthermore $\alpha'$ is compatible with $\alpha$, $F':=F^{[\alpha']}$ is a {\cuocnf} with maximum 3-clause degree $4$: We observed that critical clauses cannot appear in the process of assigning variables according to $\alpha$; furthermore any clause of $F$ not independent of $W$ must either disappear in $F'$ or become a $(\leq 2)$-clause.
As earlier, there are at most $2^{(\delb+H(\delb))n}$ cases of choosing $W$ and $\alpha'$. We now analyze what happens for the correct choice of $F'$:


We would like to use $\ws$ on $F'$; however $F'$ might contain an arbitrary amount of $(\leq 2)$-clauses. The plan is to use the fact that either there are many critical $(\leq 2)$-clauses, in which case PPSZ is better, or few critical $(\leq 2)$-clauses, in which case all other $(\leq 2)$-clauses are non-critical and have only satisfied literals.

The algorithm works as follows: We have a {\cuocnf} on $F'$ on $n':=|\vbl(F')|$ variables; the maximum degree in the 3-clauses is at most 4. First we try PPSZ: if there are $\frac{1}{30} n'$ critical $(\leq 2)$-clauses, this gives a satisfying assignment with probability  $2^{(S-0.035\frac{1}{30})n'}$. Otherwise, if there are less than $\frac{1}{10}n'$ $(\leq 2)$-clauses, the criterion of Theorem~\ref{thm.ws} applies: We invoke $\ws(F')$ with probability $2^{-(S-0.015)n'}$; this runs in expected time $2^{-o(n)}$ and finds a satisfying assignment with probability $2^{-(S-0.015)n'}$.

If both approaches fail, we know that $F'$ has at less than $\frac{1}{30} n'$ critical $(\leq 2)$-clauses clauses, but also more than $\frac{1}{10} n'$ $(\leq 2)$-clauses overall. Hence at most one third of the $(\leq 2)$-clauses is critical. However a non-critical $(\leq 2)$-clause must be a $2$-clause with both literals
 satisfied. Hence choosing a $(\leq 2)$-clause of $F'$ uniformly at random and setting all its literals to $1$ sets two variables correctly with probability at least $\frac{2}{3}>2^{-0.371\cdot 2}>2^{-(S-0.015)\cdot 2}$. That is we reduce the number of variables by 2 with a better probability than PPSZ overall; and we can
  repeat the process with the reduced formula. This shows that for the correct $F'$, we have expected running time $2^{o(n)}$ and success probability $2^{(-S+\epsc-o(1))n}$ 
  for $\epsc>0$. It is important to see that $\epsc$ does \emph{not} depend on $\delb$. Repeating this process $2^{(-S+\epsc-o(1))n}$ times gives success probability $2^{o(n)}$.
  
  Together with the brute-froce choice of $W$ and $\alpha'$, we have expected running time of $2^{(S-\epsc+\delb+H(\delb)+o(1))n}$. By choosing $\delb$ small enough we are better than PPSZ.

\end{proof}

\section{Open Problems}
\label{sec.con}

Can we also obtain an improvement for general $3$-SAT? In general $3$-SAT, there might be even fewer critical clauses and critical clauses for some assignments are not always critical for others. We need to fit our improvement into the framework of~\cite{hertli11}. As there is some leeway for multiple assignments, this seems possible, but nontrivial and likley to become very complex.

Another question is whether we can improve (Unique) $k$-SAT. PPSZ becomes slower as $k$ increases, which makes an improvement easier. However the guessing in $\sparsealg$ relied on the fact that non-critical $(\leq 2)$-clauses have all literals satisfied, which is not true for larger clauses.

Suppose Wahlstr\"om's algorithm is improved so that it runs in time $O(c^n)$ on 3-CNF formulas with average degree $D$. The sparsification lemma~\cite{ipz01} shows that for $c\to 1$ and $D\to\infty$, we obtain an algorithm for 3-SAT running in time $O(b^n)$ for $b\to 1$. Can our approach be extended to a similar sparsification result?

\paragraph{\bf{Acknowledgements}}
I am very grateful to my advisor Emo Welzl for his support in realizing this paper.


\bibliography{../common}

\newpage

\begin{appendix}
\section{Omitted Proofs}
\begin{theorem}
\label{thm.ppszbound}
Let $F$ be a {\ocnf} on $n$ variables with unique satisfying assignment $\alpha^*$. Let $x\in \vbl(F)$ and $r\in[0,1]$. In the PPSZ algorithm, conditioned on $\beta=\alpha^*$ and $\pi(x)=r$, $x$ is \emph{forced} with probability at least $b(r)$, where $b(r)=b_n(r)=\min\left\{1,\frac{r^2}{(1-r)^2}\right\}-o(1)$ is a monotone increasing function ($o(1)$ stands for a term that goes to $0$ when $n$ goes to $\infty$).
\end{theorem}

\begin{proof}[Proof of Theorem \ref{thm.ppszbound}]
This can be derived from~\cite{ppsz}. There are two differences: The first is that we define $\ppsz$ with $s$-implication instead of bounded resolution. It is easily seen that the critical clause tree construction of~\cite{ppsz} also works with $s$-implication. We use $s$-implication because we think it makes the algorithm easier to understand.

The second difference it that we give a bound for a fixed $\pi(x)$. We need this to be able to modify $\ppsz$ and analyze it in special situations. However, we can derive this result from~\cite{ppsz}: 

Let $f(r):=\min\left\{1,\frac{r^2}{(1-r)^2}\right\}$, the ``ideal'' lower bound of PPSZ that a variable is forced. Remember that $\int_0^1 f(r) = 1-S$.
In~\cite{ppsz} tehy give a lower bound $0\leq b'(r)\leq 1$ on the probability that a variable is forced given $\pi(x)$ with $b'(r)\leq f(r)$. This bound is shown to integrate to $1-S-o(1)$. As the probability that a variable is forced does not decrease if it comes later in PPSZ, the bound can easily been made monotone (if it is not already) by setting $b(r):=\max_{\rho\leq r} b'(r)$, 

For $r=0$ the statement is trivial. Now suppose for some $r\in(0,1]$, $b(r)<f(r)-\epsilon$ for some $\epsilon>0$ and all $n$. By the above observation, $f(r)-b(r)\geq 0$, and $\int_{0}^{1} (f(r)-b(r))dr=o(1)$. By continuity of $f(r)$ and monotonicity of $b(r)$, we find $r'<r$ such that $b(r)<f(r')-\epsilon$ for all $n$. But then by monotonicity of $f(r)$ and $b(r)$, $\int_0^1 (f(t)-b(t))dt \geq \int_{r'}^r (f(t)-b(t))dr \geq \int_{r'}^r(f(r')-b(r))dr > (r'-r)\epsilon$, a contradiction.
\end{proof}

To go from Theorem~\ref{thm.ppszbound}, where the place of a variable is fixed to the expectation, we need to integrate (this complicated approach gives us some flexibility later). We need the following special case of the well-known dominated convergence theorem (e.g. see~\cite{bartle11}). It essentially states that the $o(1)$ integrates to an $o(1)$ in our case.
\begin{theorem}[Dominated Convergence Theorem]
\label{thm.dc}
Let $f:[a,b]\to\R$ be a continuous function with $\int_{a}^{b} f(x)dx = t$. Let $f_n(x)=f(x)-o(1)$ be integrable with $|f_n(x)|\leq 1$. Then $\int_{a}^{b} f_n(x)dx=t-o(1)$.
\end{theorem}
Combining Theorem~\ref{thm.ppszbound} with Lemma~\ref{lem.ppsz} and the dominated convergence theorem~\ref{thm.dc} gives us the following corollary. Integrability of $f_n$ follows from monotonicity.
\begin{corollary}
\label{cor.expguessed}
Let $F$ a {\ocnf} with unique satisfying assignment $\alpha$. Then in PPSZ($F$) conditioned on $\beta=\alpha$, the expected number of guessed variables is at most $(S+o(1))n$.

Furthermore, suppose we pick every variable of $F$ with probability $p$, independently, and let $V_p$ be the resulting set. Then in PPSZ($F$) conditioned on $\beta=\alpha$, the expected number of guessed variables is at most $(S_p+o(1))n$.
\end{corollary}

\subsection{Integrals of Lemma \ref{lem.morethanone}}
\label{subs.integrals}
From Theorem~\ref{thm.ppszbound}, we know that the probability that $x$ is forced is at least $\min\left\{\frac{r^2}{(1-r)^2},1\right\}-o(1)$. Hence the overall probability that $x$ is forced is at least \[\int_{0}^1 \max\left\{2r^2-r^3,\min\left\{\frac{r^2}{(1-r)^2},1\right\}-o(1)\right\}dr\geq 0.6152-o(1).\]
The $o(1)$ integrates to $o(1)$ due to the dominated convergence theorem~(Theorem \ref{thm.dc}).
If we have a critical $(\leq 2)$-clause, the overall probability that $x$ is forced is at least \[\int_{0}^1 \max\left\{r,\min\left\{\frac{r^2}{(1-r)^2},1\right\}-o(1)\right\}dr\geq 0.649-o(1).\]

\end{appendix}

\end{document}